\pdfoutput=1
\documentclass[conference, a4paper]{IEEEtran}
\usepackage{cite}
\usepackage{indentfirst}
\usepackage{graphicx}
\usepackage{changepage}
\usepackage{stfloats}
\usepackage{amsfonts,amssymb}
\usepackage{bm}
\usepackage{algorithm}
\usepackage{algorithmic}
\usepackage{amsmath}
\usepackage{amsmath,amssymb,amsfonts}
\usepackage{times}
\usepackage{url}
\usepackage{cite}
\usepackage{textcomp}
\usepackage{xcolor}
\usepackage{color}
\usepackage{setspace}
\usepackage{enumitem}
\usepackage{float}
\usepackage{diagbox}
\usepackage[T1]{fontenc}
\usepackage[utf8]{inputenc}
\usepackage{threeparttable}
\usepackage{booktabs}
\usepackage{footnote}
\usepackage{graphicx}
\usepackage{pifont}
\usepackage{subfigure}
\usepackage{mathrsfs}
\allowdisplaybreaks[4]

\newenvironment{proof}{{\indent  \it Proof:}}{\hfill $\blacksquare$}

\begin{document}
\title{Intelligent Surface Enabled Sensing-Assisted Communication}
\author{
	
	\IEEEauthorblockN{Kaitao Meng\IEEEauthorrefmark{1}, Qingqing Wu\IEEEauthorrefmark{1},  Wen Chen\IEEEauthorrefmark{2}, and Deshi Li\IEEEauthorrefmark{3}}
	
	\IEEEauthorblockA{\IEEEauthorrefmark{1}State Key Laboratory of Internet of Things for Smart City, University of Macau, Macau, China}

	\IEEEauthorblockA{\IEEEauthorrefmark{2}Department of Electronic Engineering, Shanghai Jiao Tong University, Shanghai, China}
	
	\IEEEauthorblockA{\IEEEauthorrefmark{3}Electronic Information School, Wuhan University, Wuhan, China.} 
	
	Emails: \IEEEauthorrefmark{1}\{kaitaomeng, qingqingwu\}@um.edu.mo, \IEEEauthorrefmark{2}wenchen@sjtu.edu.cn
	\IEEEauthorrefmark{3}dsli@whu.edu.cn, 
}
\maketitle

\vspace{-10mm}
\begin{abstract}
	Vehicle-to-everything (V2X) communication is expected to support many promising applications in next-generation wireless networks. The recent development of integrated sensing and communications (ISAC) technology offers new opportunities to meet the stringent sensing and communication (S\&C) requirements in V2X networks. However, considering the relatively small radar cross section (RCS) of the vehicles and the limited transmit power of the road site units (RSUs), the power of echoes may be too weak to achieve effective target detection and tracking. To handle this issue, we propose a novel sensing-assisted communication scheme by employing an intelligent omni-surface (IOS) on the surface of the vehicle. First, a two-phase ISAC protocol, including the S\&C phase and the communication-only phase, was presented to maximize the throughput by jointly optimizing the IOS phase shifts and the sensing duration. Then, we derive a closed-form expression of the achievable rate which achieves a good approximation. Furthermore, a sufficient and necessary condition for the existence of the S\&C phase is derived to provide useful insights for practical system design. Simulation results demonstrate the effectiveness of the proposed sensing-assisted communication scheme in achieving a high throughput with low transmit power requirements.
\end{abstract}   
%
\newtheorem{thm}{\bf Lemma}
\newtheorem{remark}{\bf Remark}
\newtheorem{Pro}{\bf Proposition}
\newtheorem{theorem}{Theorem}
\newtheorem{Assum}{\bf Assumption}
\newtheorem{Cor}{\bf Corollary}

\section{Introduction}

Vehicle-to-everything (V2X) communications are expected to play an important role in next-generation wireless networks to support promising applications, such as autonomous driving and intelligent transportation \cite{Challenges2021Gyawali}. Besides high-quality information transmission, the environment sensing capability of vehicles is also of great importance due to the stringent requirements of localization accuracy and latency in V2X networks \cite{Bayesian2021Yuan}. Fortunately, the recent advances in multiple-input and multiple-output (MIMO) and millimeter-wave (mmWave) technologies offer new opportunities to simultaneously provide high-throughput and low-latency wireless communications, as well as ultra-accurate and high-resolution wireless sensing with the same wireless infrastructure \cite{liu2022integrated, Mao2021Channel}. Following these advantages, the investigation on integrated sensing and communications (ISAC) technology is well underway \cite{Liu2020Joint}.

Through the joint design of sensing and communication (S\&C) on the same infrastructure, ISAC can exploit the integration gain to achieve a flexible trade-off between S\&C \cite{Yuan2021Integrated, meng2022throughput}. Furthermore, the synergy between S\&C offers the potential to achieve the coordination gain in V2X networks, and some useful sensing-assisted communication mechanisms are proposed \cite{liu2020radar, du2021integrated}. For example, the authors in \cite{liu2020radar} proposed a novel extended Kalman filtering (EKF) framework to predict the kinematic parameters of vehicles and allocate the transmit power of the road site units (RSUs) based on S\&C requirements, thus reducing the overhead of the communication beam tracking. In \cite{du2021integrated}, the beamwidth is dynamically changed to exploit the beamforming gain of massive antennas in a more flexible way. However, the echo signals may be too weak to be utilized to detect and track the targets effectively, since the radar cross section (RCS) of the served vehicles in urban environments is generally small and the transmit power of RSUs may be limited. Moreover, high-frequency signals generally suffer from large fading loss when penetrating into vehicles, which seriously degrades the communication performance.

Recently, intelligent omni-surfaces (IOSs) have been proposed as a promising technique to achieve larger S\&C coverage and improved S\&C quality \cite{Xu2022Simultaneously}. Specifically, the signals received at the IOS can be reflected towards the incident side and/or refracted towards the other side of the IOS, thereby achieving a more flexible way to reconfigure wireless channels. Existing works on intelligent surfaces mostly focus on assisting low-mobility users, where the IOSs are generally deployed in fixed locations \cite{Simultaneously2022Mu}. However, the S\&C performance under this deployment strategy may be constrained for high-mobility vehicles due to the severe penetration loss and the limited coverage of IOSs. In a most recent work \cite{Huang2022Transforming}, an intelligent refracting surface (IRS) is deployed on the surface of a high-speed vehicle to aid its data transmission, where a training-based channel estimation technique is adopted for reliable transmission. However, the signaling overhead of channel training in \cite{Huang2022Transforming} would be large for higher positioning requirements, and the potential for communication improvement benefited from sensing needs to be further exploited. 

Based on the above discussion, we propose a novel sensing-assisted communication scheme by deploying an IOS on the surface of vehicles to enhance both the S\&C performance. In the considered system, we present a two-phase ISAC protocol, including the S\&C phase and the communication-only phase, to fulfill effective sensing-assisted communication improvements. During the S\&C phase, the parts of signals are refracted into the vehicle to improve the communication performance while others are reflected towards the RSU to assist the tracking of the vehicle, as shown in Fig.~\ref{figure1a}. During the communication-only phase, the signal power is concentrated on the communication user with the exploitation of the sensing results obtained in the former phase. In this work, the throughput maximization problem is considered by jointly optimizing the IOS phase shifts and the S\&C duration. However, solving this ISAC optimization problem is highly non-trivial since there is no closed-form expression of the achievable rate caused by the inaccurate location information of vehicles. To handle this issue, we derive a closed-form expression and present a sufficient and necessary condition to facilitate problem analysis. The main contributions of this paper can be summarized as follows:
\begin{itemize}
	\item We propose a novel IOS-aided sensing-assisted communication scheme for vehicle communication systems, where the throughput is maximized by jointly optimizing the IOS phase shifts and the S\&C duration. By presenting a two-phase protocol, the state estimation and measurement results during the S\&C phase can be effectively utilized for communication improvement.
	\item We derive a closed-form expression of the achievable rate under uncertain estimation information, based on which, a fundamental trade-off between S\&C is unveiled. We present a sufficient and necessary condition for the existence of the S\&C phase, and the problem can be simplified for the degenerated case.
\end{itemize}
\textit{Notations}: For a matrix ${\bm{X}}$,  ${\bm{X}}^\dag$, ${\bm{X}}^T$, and ${\bm{X}}^H$ respectively denote its conjugate, transpose, and conjugate transpose. 

\section{System Model and Problem Formulation}
As shown in Fig.~\ref{figure1a}, we consider an IOS-aided sensing-assisted system, where one mmWave RSU provides ISAC services for an IOS-mounted mobile vehicle. It is assumed that the RSU employs a general uniform linear array (ULA) with $M_t$ transmit antennas along the $x$-axis, and two perpendicular ULAs with $M_r$ receive antennas along the $x$- and $y$-axis, respectively, as shown in Fig.\ref{figure1a}. The vehicle is assumed to drive along a straight road that is parallel to the $x$-axis to simplify the analysis \cite{liu2020radar}. In the considered system, the motion parameters, e.g., angle, distance, and velocity of the vehicle, can be estimated by analyzing the echo signals. These parameters can be deemed as functions of time $t \in [0, T]$, with $T$ being the maximum service duration. For notational convenience, $T$ is equally divided into several small time slots, indexed by $n \in {\cal{N}} = \{1, \cdots, N\}$, and $N = {T}/{\Delta T}$. It is assumed that the motion parameters keep constant within each time slot $\Delta T$ \cite{Jayaprakasam2017Robust}. $\psi^x_{n}$ and $\psi^z_{n}$ respectively denote the azimuth and elevation angles of the geometric path between the IOS and the RSU. The Doppler frequency and the round-trip delay of echos reflected from the IOS are denoted by $\mu_{n}$ and $\tau_{n}$. 

\subsection{Channel Model}
As shown in Fig.~\ref{figure1a}, the estimated angles of the vehicle relative to the ULAs along the $x$- and $y$-axis are denoted by $\varphi_{n}$ and $\phi_{n}$, respectively, which will be transmitted to the IOS controller for phase shift design.\footnote{The estimated angles are not the azimuth and elevation angles between the IOS and the RSU, and the relationship of these angles is given below (\ref{TransmissionLink}).} A uniform planar array (UPA) is equipped at the IOS. The number of IOS elements equipped on the vehicle is denoted by $L = L_x \times L_y$, where $L_x$ and $L_y$ denote the number of elements along the $x$- and $y$-axis, respectively. The half-wavelength antenna spacing is assumed for the UPA and the ULAs. The channel between the RSU and the IOS follows  free-space pathloss models and the channel power gain from the RSU to the IOS can be expressed as $\beta_{G,n} = \beta_0 d_{n}^{-2}$, where $\beta_0$ is the channel power at the reference distance 1 m, $d_{n}$ represents the distance from the RSU to the IOS during the $n$th time slot. $\bm{H}^{\mathrm{DL}}_{n} \in \mathbb{C}^{M_t \times L}$ is the downlink channel matrix from the RSU to the IOS, given by
\vspace{-1mm}
\begin{equation}\label{TransmissionLink}
	\bm{H}^{\mathrm{DL}}_{n}=\sqrt{\beta_{G,n}} \bm{a}_{\mathrm{I}}\left(\phi_{n}, -\varphi_{n}\right) \bm{a}^T _{\mathrm{R}}\left(\varphi_{n}\right).
	\vspace{-1mm}
\end{equation}
In (\ref{TransmissionLink}), ${\bm{a}}_{\mathrm{R}}(\varphi_{n}) \!=\! [1, \!\cdots\!, e^{ {-j \pi(M_t-1) \cos(\varphi_{n}) }}]^T$, ${\bm{a}}_{\mathrm{I}}(\phi_{n}, \!-\varphi_{n}) 
		= [1, \cdots, e^{ {j \pi(L_{x}-1) \cos(\varphi_{n}) }}]^T \otimes [1, \cdots, e^{ { -j \pi(L_{y}-1) \cos({\phi}_{n}) }}]^T$, and $j$ denotes the imaginary unit. Here, $\cos(\varphi_{n}) \!= \!\sin(\psi^z_{n}) \cos (\psi^x_{n})$, $\cos({\phi}_{n}) \!= \!\sin(\psi^z_{n}) \sin (\psi^x_{n})$, and $\otimes$ denotes the Kronecker product. 
$\bm{H}^{\mathrm{UL},x}_{n}, \bm{H}^{\mathrm{UL},y}_{n} \in \mathbb{C}^{L \times M_r}$ are the uplink channel matrices from the IOS to the receive ULAs along the $x$- and $y$-axis, respectively, given by
\vspace{-2mm}
\begin{equation}
	\bm{H}^{\mathrm{UL,x}}_{n}=\sqrt{\beta_{G,n}}\bm{b}_{\mathrm{R}}\left(\phi_{n}\right) \bm{a}_{\mathrm{I}}^{T}\left(\phi_{n}, -\varphi_{n}\right), 
\end{equation}
\begin{equation} \bm{H}^{\mathrm{UL,y}}_{n}=\sqrt{\beta_{G,n}}\bm{b}_{\mathrm{R}}\left(\varphi_{n}\right) \bm{a}_{\mathrm{I}}^{T}\left(\phi_{n}, -\varphi_{n}\right),
	\vspace{-1.5mm}
\end{equation}
where ${\bm{b}}_{\mathrm{R}}\left(x\right)
= \left[1, \cdots, e^{ {-j \pi\left(M_r-1\right)  \cos(x) }}\right]^T$. Since the IOS-device channel changes much more slowly compared to the RSU-IOS channel, it is practically quasi-static and can be assumed to be an approximately line-of-sight (LoS) channel, given by ${\bm{h}}
	= \sqrt{\beta_{h}} [1, \cdots, e^{ {-j \pi(L_{x}-1)  {\Phi}^u }}]^T  \otimes [1, \cdots, e^{ { -j  \pi(L_{y}-1)  {\Omega}^u }}]^T$,
where ${\Phi}^u \triangleq \sin \left(\psi^{u,z}_{n}\right) \cos \left(\psi^{u,x}_{n}\right)$, ${\Omega}^u \triangleq \sin \left(\psi^{u,z}_{n}\right) \sin \left(\psi^{u,x}_{n}\right)$, and $\beta_{h}$ denotes the power gain of the channel from the IOS to the communication device inside the vehicle. $\psi^{u,x}_{n}$ and $\psi^{u,z}_{n}$ are respectively the azimuth and elevation angles of the geometric path connecting the IOS and the communication device. ${\bm{h}}$ can be obtained by conventional channel estimation methods \cite{Zheng2022Survey}.

\begin{figure}[t]
	\centering
	\includegraphics[width=8cm]{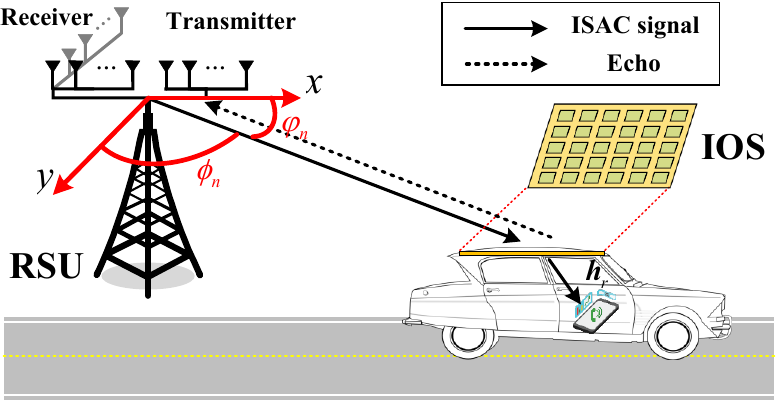}
	\vspace{-3.5mm}
	\caption{Scenarios of IOS-assisted sensing and communication.}
	\label{figure1a}
	\vspace{-2.1mm}
\end{figure}

\subsection{Proposed Sensing-assisted Communication Framework}
\label{SensingAndCommunicationModel}
To effectively balance the S\&C performance and fully exploit the improvement brought by sensing, we design a two-phase sensing-assisted communication scheme to maximize the throughput of the high-mobility user. As shown in Fig.~\ref{figure2}, each time slot is divided into two phases with a time splitting ratio of $\eta$ and $1-\eta$, where both S\&C services are provided in the former phase by splitting the power of the signals incident on the IOS into two different directions, while the latter phase is solely designed for communication by setting the IOS to be a totally refracting state. Here, the power of the refracted and reflected signals has a splitting ratio of $\beta_{l,n}^{\xi,T}$ and  $\beta_{l,n}^{\xi,R}$, where $\xi \in \{S\&C, C\}$ represents the state of the phases, i.e., the S\&C phase and the communication-only phase. Here, $\beta_{l,n}^{\xi,T} + \beta_{l,n}^{\xi,R} = 1$ and $\beta_{l,n}^{\xi,T}, \beta_{l,n}^{\xi,R} \in [0,1]$. Then, during the $\xi$ phase, the refraction- and reflection-coefficient matrices of the IOS are given by ${\bm{\Theta}}^{\xi,T}_{n} = {\rm{diag}}(\sqrt{\beta^{\xi,T}_{1,n}}e^{j \theta^{\xi,T}_{1,n}}, ... , \sqrt{\beta^{\xi,T}_{L,n}}e^{j \theta^{\xi,T}_{L,n}})$ and ${\bm{\Theta}}^{\xi,R}_{n} = {\rm{diag}}(\sqrt{\beta^{\xi,R}_{1,n}}e^{j \theta^{\xi,R}_{1,n}}, ... , \sqrt{\beta^{\xi,R}_{L,n}}e^{j \theta^{\xi,R}_{L,n}})$, where $\theta^{\xi,T}_{l,n}$ and $\theta^{\xi,R}_{l,n} \in [0, 2\pi)$, respectively denote the refracting and reflecting phase shifts of the $l$th IOS element, $l \in {\cal{L}} = \{ 1,\cdots,L\}$. 

\begin{figure}[t]
	\centering
	\includegraphics[width=8.1cm]{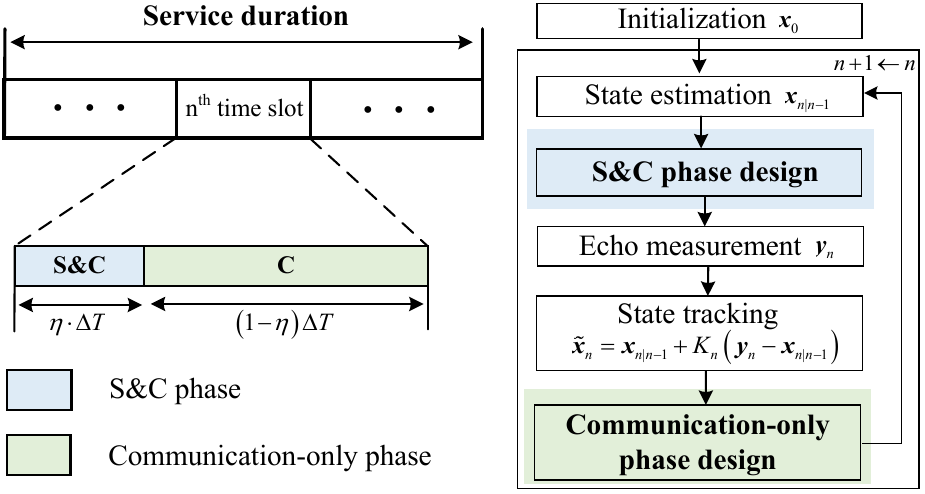}
	\vspace{-3.5mm}
	\caption{Illustration of the proposed two-phase sensing-assisted protocol.}
	\label{figure2}
	\vspace{-1mm}
\end{figure}
As illustrated in Fig.~\ref{figure2}, in the proposed protocol, during the S\&C (communication-only) phase, the RSU beamforming and IOS phase shift vectors are jointly designed based on the state estimation results (state tracking results). Specifically, the state evolution model of the vehicle can be given by \cite{liu2020radar}
\vspace{-1.5mm}
\begin{equation}
	\left\{\begin{array}{l}
		\varphi_{n}=\varphi_{n-1}+d_{n-1}^{-1} v_{n-1} \Delta T \cos \left(\varphi_{n-1}\right)+\omega_{\varphi} \\
		\phi_{n} = \phi_{n-1} + d_{n-1}^{-1}{{{v_{n - 1}}\Delta T}}\cot \left( {{\phi _{n - 1}}} \right) +\omega_{\phi}\\
		d_{n}=d_{n-1}-v_{n-1} \Delta T \sin \left(\varphi_{n-1}\right)+\omega_{d} \\
		v_{n}=v_{n-1}+\omega_{v}
	\end{array}\right. .
\vspace{-1.5mm}
\end{equation}
The covariance matrix of $\omega_{\varphi}$, $\omega_{\phi}$, $\omega_{d}$, and $\omega_{v}$ is ${\bm{Q}}_{\omega} = {\rm{diag}}([\sigma^2_{\omega_{\varphi}}, \sigma^2_{\omega_{\phi}}, \sigma^2_{\omega_{d}}, \sigma^2_{\omega_{v}}])$. Then, the state evaluation model can be expressed as ${\bm{x}}_{n|n-1}=\bm{g}\left({\bm{x}}_{n-1}\right)+\boldsymbol{\omega}_{n}$, where ${\bm{x}}_{n} = [\varphi_{n}, \phi_{n}, d_{n}, v_{n}]$. Then, Kalman filtering is adopted for beam prediction and tracking. By denoting the state estimation variables as ${\bm{x}}_{n|n-1} = [\varphi_{n|n-1}, \phi_{n|n-1}, d_{n|n-1}, v_{n|n-1}]$ and the measured signal vector as ${\bm{y}}_{n} = [\hat \varphi_{n}, \hat \phi_{n}, \hat d_{n}, \hat v_{n}]$, during the $n$th time slot, the measurement model can be formulated as $\bm{y}_{n}={\bm{x}}_{n}+\bm{z}_{n}$, where $\bm{z}_{n} \in \mathcal{C} \mathcal{N}(0, {\bm{Q}}_{z_{n}} )$ and ${\bm{Q}}_{z_{n}} = {\rm{diag}}([\sigma^2_{z_{\varphi_{n}}}, \sigma^2_{z_{\phi_{n}}}, \sigma^2_{z_{d_{n}}}, \sigma^2_{z_{v_{n}}}])$. 
Then, the prediction MSE matrix is given by 
${\bm{M}}_{n \mid n-1} \!=\! {\bm{G}}_{n-1} {\bm{M}}_{n-1} {\bm{G}}_{n-1}^H + {\bm{Q}}_{\omega}$, where $\bm{G}_{n-1}=\left.\frac{\partial \bm{g}}{\partial \bm{x}}\right|_{\bm{x}=\hat{\bm{x}}_{n-1}}$.
Then, the state tracking variable is $\tilde {\bm{x}}_{n} = [\tilde \varphi_{n}, \tilde \phi_{n}, \tilde d_{n}, \tilde v_{n}]$, given by
\vspace{-1.5mm}
\begin{equation}\label{StateTracking}
	\tilde{\bm{x}}_{n}={\bm{x}}_{n \mid n-1}+\bm{K}_{n}\left(\bm{y}_{n}-{\bm{x}}_{n \mid n-1}\right),
	\vspace{-1.5mm}
\end{equation}
where the Kalman gain is $\bm{K}_{n}=\bm{M}_{n \mid n-1} \left(\bm{Q}_z + \bm{M}_{n \mid n-1} \right)^{-1}$. 

\vspace{-1mm}
\subsection{Radar Measurement Model}
\label{RadarMeasurement}
The S\&C performance can be maximized by aligning the beam's direction towards the vehicle based on the estimation and measurement results. Hence, the transmit beamforming vectors during the S\&C phase and the communication-only phase are respectively given by ${\bm{f}}_{n}^{S\&C} = \sqrt{\frac{P^{\max}}{M_t}} {\bm{a}}^\dag_{\mathrm{R}}(\varphi_{n|n-1} ),  \ {\bm{f}}_{n}^{C} = \sqrt{\frac{P^{\max}}{M_t}} {\bm{a}}^\dag_{\mathrm{R}}(\tilde \varphi_{n} )$,
where $\varphi_{n|n-1}$ denotes the estimated angle based on the state evolution model and $\tilde \varphi_{n}$ is the angle of the state tracking (c.f. (\ref{StateTracking})). Similarly, during the S\&C phase, the receive beamforming filter for the ULAs along the $x$- and $y$-axis are respectively given by ${\bm{v}}_{n}^{x} =  \sqrt{\frac{1}{M_r}} {\bm{b}}_{\mathrm{R}}(\varphi_{n|n-1} )$ and ${\bm{v}}_{n}^{y} =  \sqrt{\frac{1}{M_r}} {\bm{b}}_{\mathrm{R}}(\phi_{n|n-1} )$.

At the $n$th time slot, the ULA along the $x$-axis receives the echoes contributed by the IOSs, expressed as ${\bm{r}}^x_{n}(t) = e^{j2\pi \mu_{n} t} {\bm{H}}^{\mathrm{UL},x}_{n} {\bm{\Theta}}_{n}^{S\&C,R} {\bm{H}}^{\mathrm{DL}}_{n} {\bm{f}}_{n}^{S\&C} s_{n}(t - \tau_{n}) + {\bm{z}}_r(t)$, where $s_{n}(t)$ is the ISAC signals with carrier frequency $f_c$, and ${\bm{z}}_r(t)$ denotes the noises at the ULAs. The resolved ranges and velocities in the delay-Doppler domain can be determined by utilizing the standard matched-filtering techniques \cite{liu2020radar}. It is assumed that the scatters can be distinguished in the delay-Doppler domain. Then, for the receive antennas along the $x$-axis, the output of the matched filter is expressed as 
\vspace{-1mm}
\begin{align}\label{MatchFiltering}
	\widetilde{\bm{r}}^x_{n} =&  \sqrt{W} ({\bm{v}}^x_{n})^{H}{\bm{H}}^{\mathrm{UL}}_{n} {\bm{\Theta}}_{n}^{S\&C,R} {\bm{H}}^{\mathrm{DL}}_{n} {\bm{f}}_{n}^{S\&C}{\delta}\left(\tau-\tau_{n}, \mu-\mu_{n}\right) \nonumber \\ &+\widetilde{\bm{z}}_{r},
\end{align}
where ${\delta}\left(\tau-\tau_{n} , \mu-\mu_{n}\right)$ is the normalized matched-filtering output function obtained by time and frequency reversing and conjugating its own waveform for the complex transmit signal $s_{n}(t)$. In (\ref{MatchFiltering}), $\widetilde{\bm{z}}_{r}$ denotes the normalized measurement noise, and $W$ denotes the matched-filtering gain. The angle $\varphi_{n}$ can be readily measured by super-resolution algorithms like multiple signal classification (MUSIC) \cite{schmidt1986multiple}. As analyzed in Section \ref{SensingAndCommunicationModel}, the AoA estimation result $\hat \varphi_{n} = \varphi_{n} + z_{\varphi_{n}}$ and $\hat \phi_{n} = \phi_{n} + z_{\phi_{n}}$, where $z_{\varphi_{n}}$ and $z_{\phi_{n}}$ represents the angle estimation errors, $ z_{\varphi_{n}} \in \mathcal{C} \mathcal{N}(0, \sigma_{z_{\varphi_{n}}}^{2} )$ and $ z_{\phi_{n}} \in \mathcal{C} \mathcal{N}(0, \sigma_{z_{\phi_{n}}}^{2} )$. The estimation error is inversely proportional to the received echo's power at the RSU \cite{Liu2022Survey}, i.e.,
\vspace{-1mm}
\begin{equation}\label{ReceivedPowerEquation}
	\sigma_{z_{\varphi_{n}}}^{2} \propto {({\gamma^{S,x}_{n}{{\sin }^2}\varphi_{n} })^{-1}}, \ \sigma_{z_{\phi_{n}}}^{2} \propto ({{\gamma^{S,y}_{n}{{\sin }^2}\phi_{n} }})^{-1},
	\vspace{-1mm}
\end{equation}
where ${\gamma^{S,x}_{n}}$ and ${\gamma^{S,y}_{n}}$ respectively denotes the signal-to-noise ratio (SNR) at the receive ULAs of the RSU along the $x$- and $y$-axis after match-filtering. Due to the uncertainty of the estimated information, the SNR of the received echos at the $x$-axis ULA is given in expectation form, i.e.,
\vspace{-1mm}
\begin{equation}\label{SensingReceivedPower}
	\begin{aligned}
			\gamma^{S,x}_{n} = &\frac{\eta \Delta T {\beta^2_{G,n}}}{\Delta t \sigma_s^2}  \mathbb{E}[ | {({\bm{v}}^x_{n})^{H}	{\bm{b}}_{\mathrm{RSU}}(\varphi_{n})}|^2 |\bm{a}^{T}_{\mathrm{IOS}}(\phi_{n},-\varphi_{n})    \\ &{\bm{\Theta}}^{S\&C,R}_{n} \bm{a}_{\mathrm{IOS}}(\phi_{n},-\varphi_{n} )|^2 |{{\bm{a}}_{\mathrm{RSU}}^T(\varphi_{n}){\bm{f}}_{n}^{S\&C}}|^{2} ],
		\vspace{-1mm}
	\end{aligned}
\end{equation}
where $\sigma_s^2$ is the noise power at the receive ULAs of the RSU, $\Delta t$ is the duration of one symbol, $\frac{\eta \Delta T}{\Delta t}$ represents the number of symbols used for filtering during the S\&C phase. 

\subsection{Communication Model}
The communication device inside each vehicle mainly receives signals via the RSU-IOS-device link since other non-line-of-sight (NLOS) links between the RSU and the communication device are practically assumed to be negligible due to severe penetration loss, especially for mmWave signals. During the S\&C phase, the signal received at the communication device can be given by $y^{S\&C}_{n}(t) =   {\bm{h}}^T{\bm{\Theta}}_{n}^{S\&C,T}\bm{H}^{\mathrm{DL}}_{n} {\bm{f}}_{n}^{S\&C} s_{n}(t) +  z(t)$, ${{z}}(t)$ denotes the noise at its receive antenna. The signal-to-noise ratio (SNR) of the user during the S\&C phase can be denoted by $\gamma^{S\&C}_{n}$, and the corresponding achievable rate is given by
\vspace{-1mm}
\begin{equation} 
	\begin{aligned}
		R^{ {S\&C }}_{n} &= \mathbb{E}\left[ \log _{2}(1+ \gamma^{S\&C}_{n})\right]  \\
		&\stackrel{(a)}{\leqslant} \log _{2}\left(1+\frac{  1}{\sigma_{c}^{2}}\mathbb{E}\left[ \left|{\bm{h}}^T{\bm{\Theta}}^{S\&C,T}_{n} {\bm{H}}^{\mathrm{DL}}_{n} {\bm{f}}_{n}^{C}\right|^{2}\right]\right) ,
		\vspace{-1mm}
	\end{aligned}
\end{equation}
where ($a$) holds based on Jensen's inequality, i.e., $\mathbb{E}\left[ \log(x) \right] \le \log\left(\mathbb{E}[x] \right)$. During the communication-only phase, the signal received at the communication device can be given by $y^C_{n}(t) =   {\bm{h}}^T{\bm{\Theta}}_{n}^{C,T}\bm{H}^{\mathrm{DL}}_{n} {\bm{f}}_{n}^C s_{n}(t) + z(t)$, where ${\bm{f}}_{n}^C$ denotes the beamforming vector of the RSU during the communication-only phase at the $n$th time slot. Similarly, the corresponding achievable rate is given by
\vspace{-1mm}
\begin{equation}
	R^C_{n} \!=\!\mathbb{E}[ \log _{2}(1+\gamma^C_{n})]  \!\le \! \log _{2}\!\left(\!1\!+\frac{1}{ \sigma_{c}^{2}} \mathbb{E}[|{\bm{h}}^T{\bm{\Theta}}_{n}^{C,T}{\bm{H}}^{\mathrm{DL}}_{n} {\bm{f}}_{n}^{C}|^{2}]\!\right)\!,
	\vspace{-1mm}
\end{equation}
where $\gamma^{C}$ denotes the SNR of the communication device. Then, during the $n$th time slot, the average achievable rate can be given by 
\vspace{-1mm}
\begin{equation}
	R_{n} = \eta R_{n}^{S \& C} + \left( 1 - \eta  \right) R_{n}^C.
	\vspace{-1mm}
\end{equation}
\subsection{Problem Formulation}

During each time slot, we aim to maximize the minimum throughput among the communication devices in the considered system by jointly optimizing the phase shift matrices of IOSs and the S\&C duration ratio. 
\vspace{-1mm}
\begin{alignat}{2}
	\label{P1}
	(\rm{P1}): \quad & \begin{array}{*{20}{c}}
		\mathop {\max }\limits_{{\bm{\Theta}}^{\xi,R}_n, {\bm{\Theta}}^{\xi, T}_n, \eta} \quad  R_n
	\end{array} & \\ 
	\mbox{s.t.}\quad
	& \theta^{\xi,T}_{l,n}, \theta^{\xi,R}_{l,n} \in [0, 2 \pi), \forall l, \xi, & \tag{\ref{P1}a}\\
	& \beta^{\xi,T}_{l,n}, \beta^{\xi,R}_{l,n} \in [0, 1], \beta^{\xi,T}_{l,n} + \beta^{\xi,R}_{l,n}  = 1, \forall  l, \xi,  & \tag{\ref{P1}b}\\
	& \eta \in [0, 1]. & \tag{\ref{P1}c}
	\vspace{-1mm}
\end{alignat}
Solving (P1) optimally is highly non-trivial due to the lack of closed-form objective function and closely coupled variables. To tackle this issue, we first derive a closed-form achievable rate achieving a good approximation, and then present some useful insights for the proposed scheme.

\section{Sensing-assisted Communication}
\label{SingleVehicle}

In this section, the optimal phase shift is given first, based on which, a closed-form expression of the achievable rate is derived.
\begin{thm}\label{OptimalPhaseShift}
	At the optimal solution of (P1), during the S\&C phase, the reflecting and refracting phase shifts, $\theta^{R}_{(l_x-1)L_y+l_y}$ and $\theta^{T}_{(l_x-1)L_y+l_y}$, of the $l_x$th-row and $l_y$th-column element of the IOS is obtained by
	\vspace{-1mm}
	\begin{equation}
		\begin{aligned}
			&\theta^{R}_{(l_x-1)L_y+l_y}=\pi(l_x-1) q^{R}_{{x}}+\pi(l_y-1) q^{R}_{{y}}+\theta_{0}, \\
			&\theta^{T}_{(l_x-1)L_y+l_y}=\pi(l_y-1) q^{T}_{{x}}+\pi(l_y-1) q^{T}_{{y}}+\theta_{0}, 
			\vspace{-1mm}
		\end{aligned}
	\end{equation}
	where {$\theta_{0}$} is the reference phase at the origin of the coordinates, $q^{R}_{{x}} =  - 2 \cos(\varphi_{n|n-1}),\  q^{R}_{{y}} =  2 \cos(\phi_{n|n-1})$, $q^{T}_{{x}} = - \cos(\varphi_{n|n-1}) + \sin ( \psi^{u,z}_{n}) \cos( \psi^{u,x}_{n})$, $q^{T}_{{y}} = \cos(\phi_{n|n-1}) + \sin ( \psi^{u,z}_{n}) \sin( \psi^{u,x}_{n})$. 
\end{thm}
\begin{proof}
	Due to page limitation, we put the corresponding proof in \cite{Myproof}, and the proofs of all subsequent lemmas and propositions in this work are put in \cite{Myproof}. 
\end{proof}

Similar to Lemma \ref{OptimalPhaseShift}, the optimal phase gradients $q_{{x}}$ and $q_{{y}}$ during the communication-only phase can be given by $q^T_{{x}} =  -\cos(\tilde \varphi_{n}) + \sin ( \psi^{u,z}_{n}) \cos( \psi^{u,x}_{n})$ and $q^T_{{y}} =  \cos(\tilde \phi_{n}) +  \sin ( \psi^{u,z}_{n}) \sin( \psi^{u,x}_{n})$.

\begin{thm}\label{EqualPowerSplitting}
	At the optimal solution of (P1), $\beta^{\xi, R*}_{l,n} = \beta^{\xi, R*}_{l',n}$, $\beta^{\xi,T*}_{l,n} = \beta^{\xi,T*}_{l',n}$, where $\xi \in \{S\&C, C\}$ and $l, l' \in {\cal{L}}$. 
\end{thm}

According to Lemma \ref{EqualPowerSplitting}, all IOS elements share the same refracted and reflected power splitting ratio, denoted by $\beta^{\xi,T}$ and $\beta^{\xi,R}$. In particular, during the communication-only phase, $\beta^{C, R} = 0$ and $\beta^{C, T} = 1$. For notational simplification, $\beta^{S\&C, R}$ and $\beta^{S\&C, T}$ are noted by $\beta^R$ and $\beta^{T}$. Then, $|{{\bm{a}}^H_{\mathrm{R}}( \varphi_n){\bm{f}}_{n}^{S\&C}}|^2$ in (\ref{SensingReceivedPower}) can be expressed as
\vspace{-1mm}
	\begin{align}
		|{{\bm{a}}^H_{\mathrm{R}}( \varphi_n){\bm{f}}_{n}^{S\&C}}|^2 =\! {P^{\max}} F_{M_t}( {\cos (\varphi_{n|n-1} ) \!-\! \cos ( \varphi_n  )} ) ,
		\vspace{-1mm}
	\end{align}
where the Fej$\acute{e}$r kernel $F_{M_t}(x) = \frac{1}{M_t}\left(\frac{\sin \frac{M_t \pi x}{2 }}{\sin \frac{\pi x}{2 }}\right)^{2}$ \cite{rust2013convergence}. Similarly, the receive beamforming gain in (\ref{SensingReceivedPower}) can be given by $\left|{\bm{v}}_{n}^{H} {{\bm{b}}_{\mathrm{R}}\left( \varphi_n\right)}\right|^2  = F_{M_r}( {\cos (\varphi_{n|n-1} ) - \cos ( \varphi_n  )} ) $. Based on Lemmas \ref{OptimalPhaseShift} and \ref{EqualPowerSplitting}, the passive beamforming gain achieved by the IOS can be given by
\begin{align}
	&\left|\bm{a}^{T}_{\mathrm{I}}\left(\phi_n, -\varphi_n\right) {\bm{\Theta}} ^{S\&C, T}_n\bm{a}_{\mathrm{I}}\left(\phi_n, -\varphi_n\right)\right|^2 \nonumber \\
 \buildrel \Delta \over = &  \beta^{R} L_x  L_y F_{L_x}\left( 2\Delta \cos \varphi_n\right) F_{L_y}\left( 2\Delta \cos \phi_n\right) ,
\end{align}
where $\beta^{R}$ represents the common reflected power ratio of all IOS elements, $\Delta \cos \varphi_n =  \cos \left(\varphi_{n|n-1} \right) - \cos \left( \varphi_n  \right)$ and $\Delta \cos \phi_n =  \cos \left(\phi_{n|n-1} \right) - \cos \left( \phi_n  \right)$. Since $\varphi _{n|n-1}$ and $\phi _{n|n-1}$ are independent of each other, we could analyze the received power at two ULAs separately. At the ULA along the $x$-axis, the SNR of the received echos can be given by
\begin{align}\label{SensingPowerAverage}
	\gamma^{S,x}_n = & \beta^{R}\eta \frac{ L  \Delta T {P^{\max}} {\beta^2_{G,n}}}{\Delta t  {\sigma_s^2} } \mathbb{E}\left[F_{L_y}\left( 2\Delta \cos \phi_n\right)  \right]  \\
	&\mathbb{E}\left[ F_{L_x}\left( 2\Delta \cos \varphi_n\right)F_{M_t}\left(\Delta \cos \varphi_n\right) F_{M_r}\left(\Delta \cos \varphi_n\right)\right]   .  \nonumber 
\end{align}
Under the perfect angle information, i.e., $\Delta \cos \varphi_n = \Delta \cos \phi_n = 0$, a power gain of order $M_t M_r L^2$ can be achieved. However, $\gamma^{S,x}_n$ in (\ref{SensingPowerAverage}) involves the expectation step, which makes it challenging to solve (P1) optimally. 

\begin{Pro}\label{InftyBand}
When $L_x \to \infty$ and $L_y \to \infty$, the SNR of the received echo can be approximated as 
\begin{equation}\label{ClosedFormSensingPower}
	\gamma^{S}_n \!\approx \! \beta^{R}\eta \frac{  \Delta T {P^{\max}}  \beta_{G,n}^2 L  M_t M_r h(\varphi_n, \sigma^2_{\omega_\varphi}) h(\phi_n, \sigma^2_{\omega_\phi}) }{{\Delta t \sigma_s^2} } ,
\end{equation}
where $h(x, y) \buildrel \Delta \over= \sum\limits_{k = -\infty}^{\infty} \frac{1}{{\sqrt {2\pi {y}} } |\sin{x}|}\!\left(\! {{e^{ - \frac{{{{2 {k^2\pi^2} }}}}{{{y}}}}} \!+\! {e^{ - \frac{{{{2\left( {(k + 1)\pi  } - x \right)}^2}}}{{{y}}}}}} \!\right)\!$.
\end{Pro}

Proposition \ref{InftyBand} presents a closed-form achievable rate when the number of IOS elements is large, which is verified to achieve a good approximation by the Monte Carlo simulations in Section \ref{simulationS}. According to the analysis in (\ref{ReceivedPowerEquation}), we have
\begin{equation}
\sigma_{z_{\varphi_{n}}}^{2} = \frac{{\sigma _R^2}}{{\gamma^S_n{{\sin }^2}\varphi_n }} =  \frac{A_{\varphi_n}}{\eta \beta^{R}},
\end{equation}
where $A_{\varphi_n} = \frac{  \Delta t \sigma_s^2 \sigma _R^2 }{{\Delta T {P^{\max}}  \beta_{G,n}^2 L  M_t M_r h(\varphi_n, \sigma^2_{\omega_\varphi}) h(\phi_n, \sigma^2_{\omega_\phi}) {{\sin }^2}\varphi_n } }$, $\sigma _R^2$ is the variance parameter, and $A_{\varphi_n}$ represents the angle variance when $\eta = 1$ and $\beta^{R} = 1$. Since $d \gg v \Delta T$, the angle variance of the state tracking error in (\ref{StateTracking}) can be given by
\vspace{-1mm}
\begin{equation}\label{StateTrackingError}
	\sigma _{\tilde \varphi_n} ^2 = {\sigma _{{\omega_{\varphi}}} ^2} - \frac{\sigma _{{\omega_{\varphi}}} ^2}{\sigma _{{\omega_{\varphi}}} ^2 + \frac{A_{\varphi_n}}{\eta \beta^{R}}} {\omega^2 _{\varphi}} = \frac{\sigma _{{\omega_{\varphi}}} ^2 A_{\varphi_n}}{\sigma _{{\omega_{\varphi}}} ^2 {\eta \beta^{R}} + {A_{\varphi_n}}}.
	\vspace{-1mm}
\end{equation}
In (\ref{StateTrackingError}), it can be seen that the variance of the state tracking error decreases monotonically with the variance of angle estimation $\sigma_{z_{\varphi_{n}}}^{2}$. Similarly, we have $\sigma _{\tilde \phi_n} ^2 = \frac{\sigma _{{\omega_{\phi}}} ^2 A_{\phi_n}}{\sigma _{{\omega_{\phi}}} ^2 {\eta \beta^{R}} + {A_{\phi_n}}}$, where $A_{\phi_n} = \frac{  \Delta t \sigma_s^2 \sigma _R^2 }{{\Delta T {P^{\max}}  \beta_{G,n}^2 L  M_t M_r h(\varphi_n, \sigma^2_{\omega_\varphi}) h(\phi_n, \sigma^2_{\omega_\phi}) {{\sin }}^2\phi_n } }$.
\begin{Pro}\label{ClosedFormAchievable}
	When $L_x \to \infty$ and $L_y \to \infty$, the achievable rates during the S\&C and communication-only phases can be respectively approximated as 
	\vspace{-1mm}
	\[\tilde R^{S\&C}_n = \log_2\left( 1 + \tilde P \left(1 - \beta^{R} \right) h(\varphi_n, \sigma^2_{\omega_\varphi}) h(\phi_n, \sigma^2_{\omega_\phi})  \right),\]
	and
	\[\tilde R^C_n = \log_2\left( 1 + \tilde P  h(\varphi_n, \sigma^2_{\tilde \varphi_n}) h(\phi_n, \sigma^2_{\tilde \phi_n})  \right),\]
where $ \tilde{ P} =  \frac{ 4P^{\max} {\beta_{G,n}} \beta_{h}   L M_t }{\sigma_{c}^{2}}$.
\end{Pro}

In Proposition \ref{ClosedFormAchievable}, the accurate angles $\varphi_n$, and  $\phi_n$ cannot be obtained, which can be practically assumed to be $\varphi_{n|n-1}$ and $\phi_{n|n-1}$, respectively. Then, the achievable rate can be rewritten as a function about $\eta$ and $\beta^R$, and the optimal solution of (P1) can be obtained via a two-dimension search over $\eta$ and $\beta^R$. To gain more insights, we will further provide some useful performance analysis to illustrate the benefits between S\&C  and provide useful insights for practical system design.

\section{Performance Analysis}

In this section, some practical approximations are adopted to draw some useful insights. In $h(x,y)$ (defined in Proposition \ref{InftyBand}), the terms $ {{e^{ - \frac{{{{2( {k\pi } )}^2}}}{{y}}}} + {e^{ - \frac{{{{2( {(k + 1)\pi } - \varphi_{n|-1} )}^2}}}{{{y}}}}}}$ is negligible when $k \ge 1$ or $k \le -1$, and thus, we have
\vspace{-1mm}
\begin{equation}
	h(x,y) \approx \frac{1}{{\sqrt {2\pi{y} }  |\sin{x}|}} \left({1 + {e^{ - \frac{{{{2\left( {\pi  - {x} } \right)}^2}}}{{y}}}}}  \right) \buildrel \Delta \over= \tilde h(x, y).
	\vspace{-1mm}
\end{equation}
This approximation is practical since the variance of the estimated angle is much less than $2\pi$. Then, the achievable rate can be approximately given by 
\begin{align}\label{AchievableRate}
&\tilde R_n \!\approx  \eta \log_2\!\left( 1 \!+ {\tilde  P} \left(1 \!- \beta^{R} \right) \tilde h(\varphi_{n|n-1}, \sigma^2_{\omega_\varphi}) \tilde h(\phi_{n|n-1}, \sigma^2_{\omega_\phi})  \!\right) \nonumber \\
&\!+\! \left(1-\eta\right) \log_2\!\left( 1 \!+ {\tilde  P}  \tilde h(\varphi_{n|n-1}, \sigma^2_{\tilde \varphi_n}) \tilde h(\phi_{n|n-1}, \sigma^2_{\tilde \phi_n})  \right).		
\end{align}

Furthermore, a condition that specifies whether it is necessary for the IOS to reflect signals for sensing is derived in the following.
For $x \in [\Delta ,\pi - \Delta]$ with ${ \frac{{2{{ \Delta }^2}}}{y}} \gg 1$, we have ${e^{ - \frac{{{{2( {\pi  - x} )}^2}}}{y}}} \ll 1$. This approximation is practical since ${{{2{{( \pi - \phi_n )}^2}}}} \gg \sigma^2_{\omega_\phi}$ holds under most general parameter setups. In this case, $\tilde R$ can be further approximated as
\vspace{-1mm}
\begin{align}
	\hat R = & \eta {\log _2}\left( 1 + C_1\left( {1 - \beta^{R} }\right)  \right) + \left( {1 - \eta } \right) \\
	& {\log _2}\left( 1 \!+\! {C_1} {\sqrt \frac{{(\sigma _{{\omega_{\varphi}}} ^2 {\eta \beta^{R}} + {A_{\varphi_n}})(\sigma _{{\omega_{\phi}}} ^2 {\eta \beta^{R}} + {A_{\phi_n}})} } {A_{\varphi_n} A_{\phi_n}}} \right), \nonumber
\end{align}
where $C_1 = \frac{ 2 {P^{\max}\beta_{G,n}} \beta_{h}    L M_t }{\pi  \sin( \varphi_{n|n-1}) \sin( \phi_{n|n-1}) \sigma_{\omega_{\varphi}} \sigma_{\omega_{\phi}} \sigma_{c}^{2}}$.
In the following, this approximated achievable rate is further analyzed to gain more useful insights for system design. In this case, (P1) is simplified to (P2).
\begin{alignat}{2}
	\label{P2}
	(\rm{P2}): \quad & \begin{array}{*{20}{c}}
		\mathop {\max }\limits_{\eta ,\beta^{R}} \quad  \hat R
	\end{array} & \\ 
	\mbox{s.t.}\quad
	& \beta^{R}\in [0,1], \eta \in [0,1] . & \tag{\ref{P2}a}
\end{alignat}

\begin{Pro}\label{SensingCondition}
	At the optimal solution of (P2), it follows that the S\&C phase is needed if and only if $\frac{{\sigma _{{\omega _\phi }}^2}}{{{A_{\phi_n} }}} + \frac{{\sigma _{{\omega _\varphi }}^2}}{{{A_{\varphi_n} }}} > 2$.
\end{Pro}

In Proposition \ref{SensingCondition}, it is shown that $\eta^* > 0$ when $\frac{{\sigma _{{\omega _\phi }}^2}}{{{A_{\phi_n} }}} + \frac{{\sigma _{{\omega _\varphi }}^2}}{{{A_{\varphi_n} }}} > 2$; otherwise, $\eta^* = 0$. Intuitively, the derived sufficient and necessary condition in Proposition \ref{SensingCondition} tends to be satisfied when the estimation accuracy is lower or the measurement accuracy is higher. 

\section{Simulation Results}
\label{simulationS}
\begin{table}[b]
	\footnotesize
	\centering
	\caption{System Parameters}
	\label{tab1}
	\begin{IEEEeqnarraybox}[\IEEEeqnarraystrutmode\IEEEeqnarraystrutsizeadd{2pt}{1pt}]{v/c/v/c/v/c/v/c/v/c/v/c/v}
		\IEEEeqnarrayrulerow\\
		&\mbox{Parameter}&&\mbox{Value}&&\mbox{Parameter}&&\mbox{Value}&&\mbox{Parameter}&&\mbox{Value}&\\
		\IEEEeqnarraydblrulerow\\
		\IEEEeqnarrayseprow[2pt]\\
		&T&&10 \mbox{s}&&N&&500&&M_t&&8&\IEEEeqnarraystrutsize{0pt}{0pt}\\
		\IEEEeqnarrayseprow[2pt]\\
		\IEEEeqnarrayrulerow\\
		\IEEEeqnarrayseprow[2pt]\\  
		&M_r&& 8 &&L_x&&80&&L_y&&80&\IEEEeqnarraystrutsize{0pt}{0pt}\\
		\IEEEeqnarrayseprow[2pt]\\
		\IEEEeqnarrayrulerow\\
		\IEEEeqnarrayseprow[2pt]\\  
		&\sigma_R^2&& 10^{-1} &&P^{\max}&&0.1  \mbox{W}&&\beta_0&& -30 \mbox{dB} &\IEEEeqnarraystrutsize{0pt}{0pt}\\
		\IEEEeqnarrayseprow[2pt]\\
		\IEEEeqnarrayrulerow\\
		\IEEEeqnarrayseprow[2pt]\\
		&\sigma_s^2&&-70 \mbox{dBm}&&\sigma^2_{\omega_\phi} \ \sigma^2_{\omega_\varphi}&&0.1 &&\Delta t&&0.1 \mu\mbox{s} &\IEEEeqnarraystrutsize{0pt}{0pt}\\
		\IEEEeqnarrayseprow[2pt]\\
		\IEEEeqnarrayrulerow\\
		\IEEEeqnarrayseprow[2pt]\\
		&\sigma_c^2&& -70 \mbox{dBm} &&f_c && 30 \mbox{GHz} && \Delta T && 0.02 \mbox{s} &\IEEEeqnarraystrutsize{0pt}{0pt}\\
		\IEEEeqnarrayseprow[2pt]\\
		\IEEEeqnarrayrulerow
	\end{IEEEeqnarraybox}  
\end{table} 

In this section, Monte Carlo simulation results are provided for characterizing the performance of the proposed sensing scheme and for gaining insight into the design and implementation of IOS-assisted ISAC systems. The main system parameters are listed in Table \ref{tab1}. The coordinate of the RSU is (0 m, 0 m, 20 m), the initial location of the vehicle is (-100 m, 20 m, 0 m), and the vehicle's speed is 20 m/s. We compare the proposed schemes with two benchmarks: 
\begin{itemize}
	\item {\bf{Refraction}}: The reflecting phase shifts ${\bm{\Theta}}^{\xi,R}_{n}$ are designed randomly with the optimal $\eta$, $\beta^R$, and ${\bm{\Theta}}^{\xi,T}_{n}$.
	\item {\bf{Prediction}}: The beamforming and phase shifts are designed only based on the state estimation, i.e., $\eta$ = 0. 
\end{itemize}

\begin{figure*}[t]
	\begin{minipage}[t]{0.25\linewidth}
		\centering
		\includegraphics[width=4.4cm]{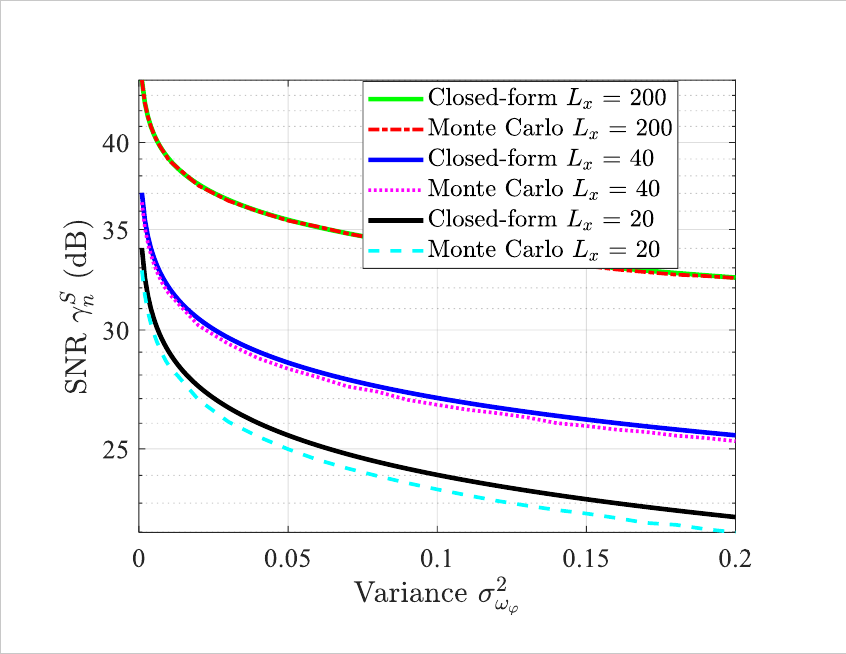}
		\caption{Evaluation of the Closed- \\
			form expression.}
		\label{figure6a}
	\end{minipage}%
	\begin{minipage}[t]{0.25\linewidth}
		\centering
		\includegraphics[width=4.4cm]{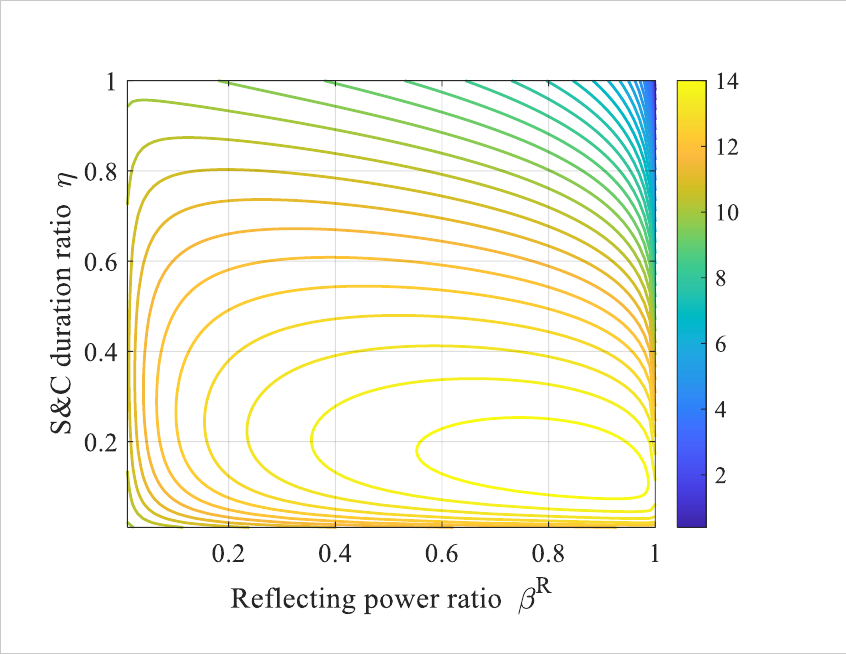}
		\caption{Impact of $\beta^R$ and $\eta$ on the \\
			 achievable rate.}
		\label{figure6b}
	\end{minipage}%
	\begin{minipage}[t]{0.25\linewidth}
		\centering
		\includegraphics[width=4.4cm]{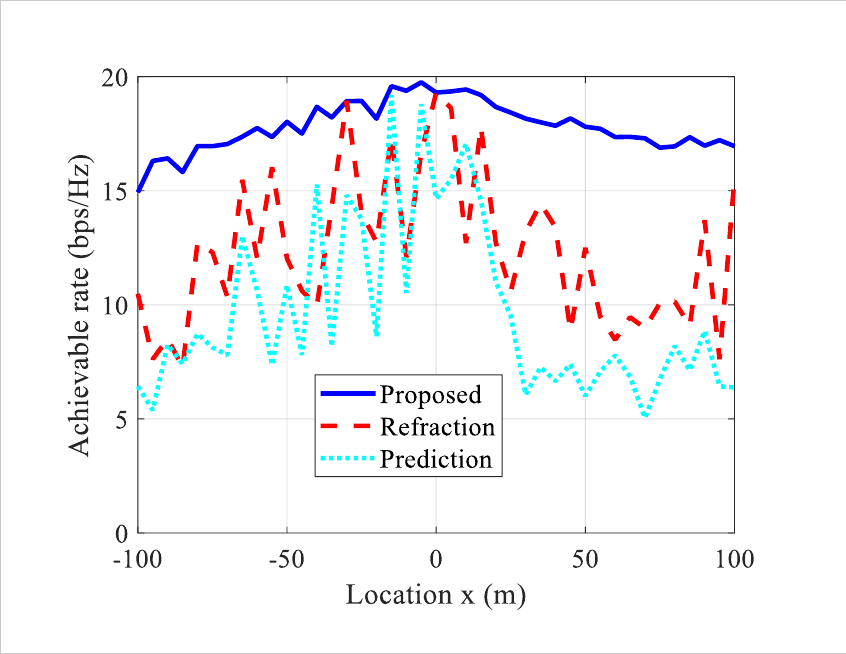}
		\caption{Comparison of achievable \\
			rates versus locations.}
		\label{figure6c}
	\end{minipage}%
	\begin{minipage}[t]{0.25\linewidth}
		\centering
		\includegraphics[width=4.4cm]{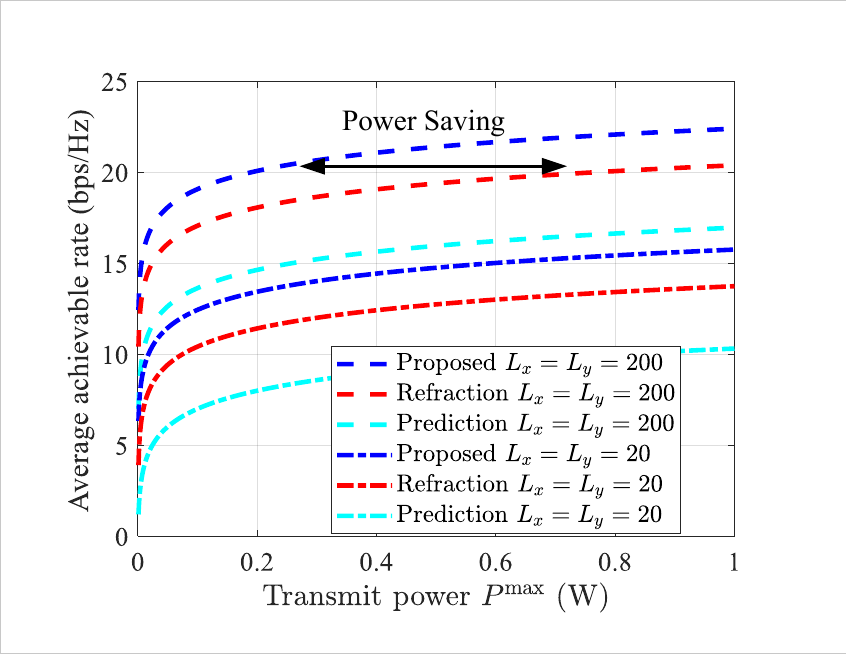}
		\caption{Comparison of the achievable rate under different transmit power.}
		\label{figure6d}
	\end{minipage}%
\end{figure*}

As shown in Fig.~\ref{figure6a}, the SNR $\gamma^S_n$ of Monte Carlo simulations with $L_x = 40$ approaches to the closed-form expression of the SNR derived in Proposition \ref{InftyBand}, which confirms that it achieves a good approximation. In Fig.~\ref{figure6b}, at the optimal solution, $\eta^* = 0.18$ and $\beta^{R*} = 0.8$. It can be shown that with any given $\beta^R$, the achievable rate first increases and then decreases with the increase of the S\&C duration ratio $\eta$. The main reason is that appropriately increasing the sensing duration can reduce the variance of parameter measurements and improve the beamforming gain, however, the communication performance gain brought by sensing cannot compensate for the performance loss caused by sensing time and power consumption during the S\&C phase, especially for the ratio $\eta$ approaching to one.

In Fig.~\ref{figure6c}, it can be seen that the achievable rate of the benchmark schemes fluctuates greatly since the state estimation error of the vehicle location is relatively large, leading to a significant performance loss caused by beam misalignment of the RSU and the IOS. Our proposed method provides higher throughput and stable data transmission for the communication user. The main reason is that both the state estimation and measurement results are exploited to improve the beam alignment performance, and the optimal ratio of sensing time and reflecting power is taken to achieve a better tradeoff between S\&C performance, thereby leading to an effective communication performance improvement benefiting from sensing.

Fig.~\ref{figure6d} shows the performance comparison among the transmit power, the number of IOS elements, and the average achievable rate. Specifically, when the transmit power threshold ${{P}^{\max}}$ is less than 0.01 W, the achievable rate of our proposed scheme is significantly improved compared to the benchmark schemes. Moreover, to achieve the same achievable rate, the proposed scheme can significantly reduce the transmit power requirements, and the main reason is that proper power and time allocation for sensing can effectively improve the beam alignment, thereby providing a large beamforming gain of the RSU and the IOS within the communication-only duration. 

\section{Conclusions}
\label{Conclusion}
In this paper, we investigated a new IOS-enabled sensing-assisted communication system and proposed a two-phase ISAC protocol to fulfill efficient balance between S\&C. The IOS phase shift and the sensing duration were jointly optimized to maximize the achievable rate. A closed-form expression of the achievable rate under uncertain locations was derived to facilitate resource allocation. Furthermore, a sufficient and necessary condition for the existence of the S\&C phase is derived to further simplify the problem. The numerical results validated the efficiency of our design over the benchmark schemes and also confirmed the benefits of the sensing-assisted communication framework.

\footnotesize  	
\bibliography{mybibfile}
\bibliographystyle{IEEEtran}

\normalsize 
\section*{Appendix A: \textsc{Proof of Lemma \ref{OptimalPhaseShift}}}
First, it can be readily proved that the phase shift optimization problem for maximizing the SNR of the received echoes is equivalent to the problem of the IOS beamforming gain maximization, i.e.,
\begin{equation}
	\mathop {\max }\limits_{{\bm{\Theta}}^{\xi, R}_{n}} \  |{\bm{a}^{T}_{\mathrm{I}}\left(\phi_{n},-\varphi_{n}\right) {\bm{\Theta}}^{\xi, R}_{n} \bm{a}_{\mathrm{I}}\left(\phi_{n},-\varphi_{n} \right)}|^2.
\end{equation}
We have ${\bm{a}^{T}_{\mathrm{I}}(\phi_{n},-\varphi_{n}) {\bm{\Theta}}^{\xi, R}_{n} \bm{a}_{\mathrm{I}}(\phi_{n},-\varphi_{n} )} = {\bm{\theta}}^{\xi, R}_{n} {\bm{a}}_{\mathrm{I}}\left(\phi_{n},-\varphi_{n}\right) \odot {\bm{a}}_{\mathrm{I}}\left(\phi_{n},-\varphi_{n}\right)$ according to $\mathbf{A}^T {\bm{X}} \mathbf{B}={\bm{x}}^T(\mathbf{A} \odot \mathbf{B})$ with ${\bm{X}} = {\rm{diag}}({\bm{x}})$, where $\odot$ denotes the Hadamard product.
The received power at the RSU is maximized when the echo reflected from the IOS is directed towards the RSU, i.e., $\theta^{R}_{(l_x-1)L_y+l_y}= - 2 \pi(l_x-1)  \cos(\varphi_{n|n-1})+ 2 \pi(l_y-1) \cos(\phi_{n|n-1})+\theta_{0}$, where {$\theta_{0}$} is the reference phase at the origin of the coordinates. Similarly, for the phase shift optimization problem of maximizing the SNR of the communication device, the optimal phase shift can be given by  $\theta^{T}_{(l_x-1)L_y+l_y}=\pi(l_y-1) (\sin ( \psi^{u,z}_{n}) \cos( \psi^{u,x}_{n})- \cos(\varphi_{n|n-1}) )+\pi(l_y-1)(\cos(\phi_{n|n-1}) + \sin ( \psi^{u,z}_{n}) \sin( \psi^{u,x}_{n}))+\theta_{0}$. Thus, the proof is completed.

\normalsize 
\section*{Appendix B: \textsc{Proof of Lemma \ref{EqualPowerSplitting}}}

According to Lemma \ref{OptimalPhaseShift}, the phase shifts of all IOS elements are aligned. It is assumed that at the optimal solution of (P1), the total reflected energy satisfies $\sum\nolimits_{l = 1}^L { {\beta _{l,n}^{\xi,R}} } = X_R$. According to Cauchy-Buniakowsky-Schwartz inequality, we have
\begin{align}\label{IOSemelemtsEquality}
	{( {\sum\nolimits_{l = 1}^L {\sqrt {\beta _{l,n}^{\xi,R}} } } )^2} \!\le\! &  \sum\nolimits_{l = 1}^L {\beta _{l,n}^{\xi,R}}  \!+\! \sum\nolimits_{l = 1}^L \sum\nolimits_{l' \ne l}^L ( {\beta _{l,n}^{\xi,R} + \beta _{l',n}^{\xi,R}} ) \nonumber \\
	=& L\sum\nolimits_{l = 1}^L {\beta _{l,n}^{\xi,R}}   = L X_R,
\end{align}
where the equality in (\ref{IOSemelemtsEquality}) holds if and only if $\beta_l^{\xi,R} = \beta_{l'}^{\xi,R}$, $\forall l, l' \in {\cal{L}}$. Similarly, $\beta_l^{\xi,T} = \beta_{l'}^{\xi,T}$. Thus, the proof is completed.

\normalsize 
\section*{Appendix C: \textsc{Proof of Proposition \ref{InftyBand}}}
\begin{figure*}[b]
	\begin{align}\label{ProbablityEquation}
		&P_\varphi(y) \!=\! \sum\limits_{k = -\infty}^{\infty} {G}{\left( {2(k + 1)\pi  -  {\arccos \left( {{{\frac{y}{2}}} + \cos \left( {{\varphi_n}} \right)} \right) \!-\! {\varphi_n}}} \right)^\prime }\!-\! {G}{\left( {2k\pi  \!+\! \arccos \left( {{\frac{y}{2}} + \cos \left( {{\varphi_n}} \right)} \right) \!-\! {\varphi_n}} \right)^\prime }  \nonumber \\
		&\!=\! \sum\limits_{k = -\infty}^{\infty} \frac{1}{{\sqrt {1 \!-\! {{\left( {{\frac{y}{2}} \!+\! \cos \left( {{\varphi_n}} \right)} \right)}^2}} }}\frac{1}{2{\sqrt {2\pi } {\sigma_{\omega_\varphi}}}}\!\left(\! {{e^{ - \frac{{{{\left( {2(k + 1)\pi  -  {\arccos \left( {{\frac{y}{2}} + \cos \left( {{\varphi_n}} \right)} \right)}  - {\varphi_n}} \right)}^2}}}{{2{\sigma^2_{\omega_\varphi}}}}}} + {e^{ - \frac{{{{\left( {2k\pi  + \arccos \left( {y + \cos \left( {{\varphi_n}} \right)} \right)} \!-\! \varphi_n \right)}^2}}}{{2{\sigma^2_{\omega_\varphi}}}}}}} \!\right).\!
	\end{align}
\end{figure*}
Before proving Proposition \ref{InftyBand}, we first introduce a Lemma to facilitate the derivation of the closed-form achievable rate.
\begin{thm}\label{FejerKernel}
	Let $\sigma_{M}(x)=\frac{1}{2 a} \int_{-a}^{a} g(x+u) \tilde F_{M}(u)(u) d u$, where $\tilde F_{M}(u) = \frac{1}{M}\left(\frac{\sin \frac{N \pi u}{2a}}{\sin \frac{\pi u}{2a}}\right)^2$. If $g(x)$ is a real valued, continuous function with period $2a$, $\sigma_M(x)$ converges uniformly to $g(x)$ when $M \to \infty$, i.e.,
	\begin{equation}
		\frac{1}{2 a} \int_{-a}^{a} g(x+u) \tilde F_{M}(u) d u-g(x) \frac{1}{2 a} \int_{-a}^{a} \tilde F_{M}(u) d u = 0.
	\end{equation}
\end{thm}
\begin{proof}
	For a detailed proof, please refer to Theorem 3.6 in \cite{rust2013convergence}.
\end{proof}

According to Lemma \ref{FejerKernel}, we have $\lim _{M \rightarrow \infty} \frac{1}{2 a} \int_{-a}^a g(x+u) \tilde F_{M}(u) d u \rightarrow g(x)$, and thus, it follows that 
\begin{equation}\label{EquationLemmaFejer}
	\lim _{M \rightarrow \infty} \frac{1}{2 a} \int_{-a}^a g(u) \tilde F_{M}(u) d u \rightarrow g(0).
\end{equation}
Furthermore, let $y = 2\cos \left( {{\varphi_{n|n-1}}} \right) - 2\cos \left( \varphi_n  \right)$, then the probability density function (PDF) of $y$ can be given in (\ref{ProbablityEquation}), as shown at the bottom of this page, where ${G}(\cdot)$ is the Cumulative distribution function (CDF) of Gaussian function.
Define \[H(y) \!=\! \left\{\! {\begin{array}{*{20}{c}}
		\!{{F_{{M_t}}}( {\frac{y}{2}} ){F_{{M_r}}}( {\frac{y}{2}} ){P_\varphi }(y),\!}&{\!y \in [ { \!- 1 \!-\! \cos {\varphi _n},1 \!-\! \cos {\varphi _n}} ]}\\
		0&{{\rm{otherwise}}}
\end{array}} \right.\]
and $\tilde{H}(y) = \sum\nolimits_{i = -\infty}^{\infty} H(y - 2i)$, where $\tilde{H}(y)$ is real valued, continue function with period $2$. Then, it follows that
\begin{equation}
\begin{aligned}
	&\mathbb{E}\left[ F_{L_x}\left( 2\Delta \cos \varphi_n\right)F_{M_t}\left(\Delta \cos \varphi_n\right) F_{M_r}\left(\Delta \cos \varphi_n\right)\right]	\\
	= &	\int\limits_{ - 2 - 2\cos {\varphi_n}}^{2 - 2\cos {\varphi_n}} {  \left( F_{L_x}\left( y\right)F_{M_t}\left(\frac{y}{2}\right) F_{M_r}\left(\frac{y}{2}\right)\right) P_\varphi(y)dy} \\
	\stackrel{(b)} \approx & \int\limits_{ - 1 - \cos {\varphi_n}}^{1 - \cos {\varphi_n}} { \tilde{H}(y)   \tilde F_{L_x}\left( y\right)dy} \\
	\stackrel{(c)}= & 2 \tilde{H}\left(0\right) = 2 M_t M_r P_\varphi(0),
\end{aligned}
\end{equation} 
where $F_{M}(x) = \frac{1}{M}\left(\frac{\sin \frac{M \pi x}{2 }}{\sin \frac{\pi x}{2 }}\right)^{2}$ and $P_\varphi(0) = \sum\limits_{k = -\infty}^{\infty} \frac{1}{2{\sqrt {2\pi {\sigma^2_{\omega_\varphi}}}|\sin({\varphi_n})| }}\!\left(\! {{e^{ - \frac{{{{\left( {2k\pi} \right)}^2}}}{{2{\sigma^2_{\omega_\varphi}}}}}} + {e^{ - \frac{{{{\left( {2(k + 1)\pi  } - 2\varphi_n \right)}^2}}}{{2{\sigma^2_{\omega_\varphi}}}}}}} \!\right)\!$, ($b$) holds since $H(y) \approx 0$ for $y \notin [-1 - \cos(\varphi_n), 1- \cos(\varphi_n)]$, and ($c$) holds based on (\ref{EquationLemmaFejer}), i.e., let $a = 1$ in Lemma \ref{FejerKernel}. Similarly, $\mathbb{E}\left[ F_{L_y}\left( 2\Delta \cos \phi_n\right)\right] \approx  \sum\limits_{k = -\infty}^{\infty} \frac{1}{{\sqrt {2\pi {\sigma^2_{\omega_\phi}}} } |\sin({\phi_n})|}\left( {{e^{ - \frac{{{{\left( {2k\pi} \right)}^2}}}{{2{\sigma^2_{\omega_\phi}}}}}} + {e^{ - \frac{{{{\left( {2(k + 1)\pi  } - 2\phi_n \right)}^2}}}{{2{\sigma^2_{\omega_\phi}}}}}}} \right)$, and thus  the expectation of the received echo's power at the ULA along the $x$-axis can be given by 
\begin{equation}
	\gamma^{S,x}_n = \beta^{R}\eta \frac{ L  \Delta T {P^{\max}}  \beta_{G,n}^2  M_t M_r h(\varphi_n, \sigma^2_{\omega_\varphi}) h(\phi_n, \sigma^2_{\omega_\phi}) }{{\sigma_s^2} } ,
\end{equation}
where 
$h(x, y) \buildrel \Delta \over= \sum\limits_{k = -\infty}^{\infty} \frac{1}{{\sqrt {2\pi {y}} |\sin({x})|}}\!\left(\! {{e^{ - \frac{{{{2( {k\pi} )}^2}}}{{{y}}}}} \!+\! {e^{ - \frac{{{{2\left( {(k + 1)\pi  } - x \right)}^2}}}{{{y}}}}}} \!\right)\!$. Similarly, it can be readily proved that $\gamma^{S,y}_n = \gamma^{S,x}_n$.

\normalsize 
\section*{Appendix D: \textsc{Proof of Proposition \ref{ClosedFormAchievable}}}
According to Lemmas \ref{OptimalPhaseShift} and \ref{EqualPowerSplitting}, by adopting a method similar to the proof of Proposition \ref{InftyBand}, the SNR at the communication device during the S\&C phase can be approximately given by
\begin{equation}\label{AchievableRateCommunication}
	\begin{aligned}
		&\gamma^{S\&C}_n \approx \mathbb{E}\left[\frac{  \left|{\bm{h}}^T{\bm{\Theta}}^{S\&C,T}_{n}\bm{H}^{\mathrm{DL}}_{n} {\bm{f}}_{n}^{S\&C}\right|^{2}}{\sigma_{c}^{2}}\right] \\
		=& \left(1 - \beta^{R} \right)\frac{ 4P^{\max}  {\beta_{G,n}} \beta_{h}   L M_t }{\sigma_{c}^{2}} h(\varphi_n, \sigma^2_{\omega_\varphi}) h(\phi_n, \sigma^2_{\omega_\phi}).
	\end{aligned}
\end{equation}
During the communication-only phase, all the signal will be refracted towards the direction of the device for communication improvement, i.e., $\beta^{C,R} = 0$ and $\beta^{C,T} = 1$, and the corresponding SNR can be approximately given by
\begin{equation}
	\begin{aligned}
		\gamma^{C}_n \approx & \mathbb{E}\left[\frac{  \left|{\bm{h}}^T{\bm{\Theta}}^{C,T}_{n}\bm{H}^{\mathrm{DL}}_{n} {\bm{f}}_{n}^{S\&C}\right|^{2}}{\sigma_{c}^{2}}\right] \\
		=& \frac{ 4P^{\max}  {\beta_{G,n}} \beta_{h}   L M_t }{\sigma_{c}^{2}} h(\varphi_n, \sigma^2_{\tilde \varphi_n}) h(\phi_n, \sigma^2_{\tilde \varphi_n})  .
	\end{aligned}
\end{equation}
Thus, the proof is completed.

\normalsize 
\section*{Appendix E: \textsc{Proof of Proposition \ref{SensingCondition}}}
\label{EqualProblemP1}

It is not difficult to verify that with any given $\eta$, the achievable rate $\hat R$ is concave about $\beta$. Similarly, with any given $\beta^{R}$, the achievable rate $\hat R$ is concave about $\eta$. For ease of analysis, $\hat R$ is noted as a function of $\eta$ and $\beta^R$, i.e., 
\begin{equation}
	\begin{aligned}
		\hat R(\eta ,\beta^{R} )  =& \eta {\log _2}\left( {1 +  C_1 {\left(1-\beta ^R\right)}} \right) + ( {1 - \eta } ) \\
		&{\log _2}\left( {1 + D_1\sqrt {(\eta \beta ^{R} + D_2)(\eta \beta ^{R} + D_3)} } \right)		,
	\end{aligned}
\end{equation}
where $D_1 = {{C_1}\sqrt {\frac{{\sigma _{{\omega _\varphi }}^2\sigma _{{\omega _\phi }}^2}}{{{A_\varphi }{A_{\phi_n} }}}} }$, $D_2 = \frac{{{A_\varphi }}}{{\sigma _{{\omega _\varphi }}^2}}$, and $D_3 = \frac{{{A_{\phi_n} }}}{{\sigma _{{\omega _\phi }}^2}}$. If $ \left. {\frac{\partial \hat R(\eta ,\beta^{R} )}{\partial \eta}} \right|_{\eta=0} \le 0$, we have $\hat R(\eta ,\beta^{R} ) \le \hat R(0 ,\beta^{R} )$, and in this case, at the optimal solution of (P2), $\eta^* = 0$.
\begin{equation} 
	\begin{aligned}
		&\left. {\frac{\partial \hat R(\eta ,\beta^{R} )\!}{\partial \eta}} \right|_{\eta=0}\! \!= \! {\log _2}( {1 \!+\! C_1 {(1-\beta ^R)}} ) \!-\! {\log _2}( {1 \!+\! D_1\sqrt {{D_2}{D_3}} } )  \\
		&+ \frac{{D_1\left( {{D_2} + {D_3}} \right){\beta ^R}}}{{2\sqrt {{D_2}{D_3}} \left( {1 + D_1\sqrt {{D_2}{D_3}} } \right)\ln 2}} \buildrel \Delta \over = g(\beta ^R).	
	\end{aligned}
\end{equation}
It can be verified that $g(\beta ^R)$ is concave about $\beta^R$. If $g'(0) \le 0$, $g(\beta ^R) \le g(0) = 0$, i.e., $ \left. {\frac{\partial \hat R(\eta ,\beta^{R} )}{\partial \eta}} \right|_{\eta=0} \le 0$ always holds. When $g'(0) \le 0$, $\frac{{ - {C_1}}}{1+C_1} + \frac{{D_1\left( {{D_2} + {D_3}} \right)}}{{2\sqrt {{D_2}{D_3}} \left( {1 + D_1\sqrt {{D_2}{D_3}} } \right)}} \le 0$, i.e., $\frac{{\sigma _{{\omega _\phi }}^2}}{{{A_{\phi_n} }}} + \frac{{\sigma _{{\omega _\varphi }}^2}}{{{A_{\varphi_n} }}} \le 2$, in this case, $\eta^* = 0$. Otherwise, there always exist $\beta^R$ making $g(\beta ^R) > 0$, i.e., $\hat R(\eta ,\beta^{R} ) > \hat R(0 ,\beta^{R} )$ for $\eta > 0$. In this case, at the optimal solution of (P2), $\eta^* > 0$. Thus, the S\&C phase is needed if and only if $\frac{{\sigma _{{\omega _\phi }}^2}}{{{A_{\phi_n} }}} + \frac{{\sigma _{{\omega _\varphi }}^2}}{{{A_{\varphi_n} }}} > 2$.

\end{document}